%% file: main_AAAI14.tex
\documentclass[letterpaper]{article}
\usepackage{aaai}
\usepackage{times}
\usepackage{helvet}
\usepackage{courier,color}
\usepackage[english]{babel}
\newtheorem{theorem}{Theorem}

\newtheorem{claim}{Claim}

\newcommand{\blackslug}{\penalty 1000\hbox{
    \vrule height 8pt width .4pt\hskip -.4pt
    \vbox{\hrule width 8pt height .4pt\vskip -.4pt
          \vskip 8pt
      \vskip -.4pt\hrule width 8pt height .4pt}
    \hskip -3.9pt
    \vrule height 8pt width .4pt}}

\newenvironment{proof}{$\;$\newline \noindent {\sc Proof.}$\;\;\;$\rm}{\qed}
\newcommand{\qed}{\hspace*{\fill}\blackslug}

\def\boxit#1{\vbox{\hrule\hbox{\vrule\kern4pt
  \vbox{\kern1pt#1\kern1pt}
\kern2pt\vrule}\hrule}}

\begin{document}
\title{How Credible is the Prediction of a Party-Based Election?}
\author{Jiong Guo \and Yash Raj Shrestha \and Yongjie Yang\thanks{Corresponding author (yyongjie@mmci.uni-saarland.de)}\\
Universit\"{a}t des Saarlandes \\
Campus E 1.7, D-66123 Saarbr\"{u}cken, Gemarny}

\nocopyright
\maketitle

\begin{abstract}
\begin{quote}
In a party-based election system, the voters are grouped into parties and all voters of a party are assumed to vote according to the party preferences over the candidates. Hence, once the party preferences are declared the outcome of the election can be determined. However, in the actual election, the members of some ``instable" parties often leave their own party to join other parties. We introduce two parameters to measure the credibility of the prediction based on party preferences: {\sc{Min}} is the minimum number of voters leaving the instable parties such that the prediction is no longer true, while {\sc{Max}} is the maximum number of voters leaving the instable parties such that the prediction remains valid. Concerning the complexity of computing {\sc{Min}} and {\sc{Max}}, we consider both positional scoring rules (Plurality, Veto, $r$-Approval and Borda) and Condorcet-consistent rules (Copeland and Maximin). We show that for all considered scoring rules, {\sc{Min}} is polynomial-time computable, while it is NP-hard to compute {\sc{Min}} for Copeland and Maximin. With the only exception of Borda, {\sc{Max}} can be computed in polynomial time for other scoring rules. We have NP-hardness results for the computation of {\sc{Max}} under Borda, Maximin and Copeland.
\end{quote}
\end{abstract}

\section{Introduction}
Voting has been recognized as a common approach for preference aggregation and collective decision making whenever there exists more than one
alternative for a community to choose from. Based on the conflicting preferences over the alternatives of different voters, some voting rules are designed
in an effort to reach the best possible joint decision. Since long, voting has been a part and parcel of the fields of preference handling, decision making
and social choice. It comes with a  wide variety of applications which ranges from
multi-agent systems, political elections, recommendation systems, etc. \cite{DBLP:journals/cj/PittKSA06,DBLP:conf/hci/Popescu13b}.

By the celebrated Gibbard-Satterthwaite theorem
\cite{Gibbard73,Satterthwaite75} and other results expanding its score (see, e.g, the work by Duggan and Schwarthz \cite{bbb})
all reasonable voting systems are manipulable in principle, as long as they are not  restricted to extremely
special cases, e.g., the single-peaked restriction. This gives rise to the possibility for voters to misreport their preferences in order
to get better off. Motivated by this fact,
Bartholdi, Tovey and Trick \cite{BARTHOLDI89} initiated the study of this issue from the complexity theoretic aspect
with the adoption of computational complexity shields as the natural barrier to prevent instability of voting.
Their seminal work paved the path for huge amount of research work exploring the complexity of various strategic behaviors (e.g., manipulation, control, bribery) in voting systems
which has been
extensively studied in the past two decades \cite{DBLP:conf/aaai/ParkesX12,DBLP:conf/atal/FaliszewskiHH13,DBLP:conf/sofsem/ChevaleyreELM07}.

In this paper, we study the voting systems, where voters can be grouped into parties (or interest groups) and the party members
are required to follow party discipline, that is, the voters of the same party should all vote according to the party preference. In this setting, the outcome of the voting can be easily predicted, once the preferences of all parties are declared. Such voting scenarios can be found in various real-world applications, for example, parliament voting and board elections of universities.

However, in practice, the final results of such elections are often much different from the predictions based on party preferences, mainly caused by the ``instability'' of some participating parties. That is, some members of these `` instable'' parties refuse to follow the preferences of their own parties and join other ``stable'' parties, possibly persuaded by the stable parties. Thus, it could be of great importance for the chairman of the voting to measure the influence of the instable parties to the predictability of the voting. Hereby, consider the following two parameters: {\sc min} represents the minimum number of voters from the instable parties, who can change the outcome of the voting by joining other stable parties, and {\sc  max} represents the maximum number of voters from the instable parties, whose revoting will not affect the outcome. Based on these definitions the prediction of such a party-based voting with high values of both {\sc min} and {\sc max} can be considered as credible. These two parameters could be also critical for party leaders to design their strategy for manipulating the outcome of the election. For example,  for a party fearing an unfavorable prediction, the parameter {\sc min} indicates the minimum ``budget'' that the party needs to invest, that is, to persuade how many voters from instable parties, while
the minimum goal for the parties favoring the prediction is to have $|\mathcal{V}|$-{\sc max} voters obeying their own party preferences, where $\mathcal{V}$ is the set of all voters.

The main task of this work is to explore the computational complexity of computing {\sc min} and {\sc max} for various voting rules. We show that {\sc min} is polynomial-time computable for all common positional scoring rules and for Condorcet, but NP-hard for Maximin and Copeland. Moreover, the computation of Max can be done in polynomial-time for $r$-Approval, Plurality and Veto, but is NP-hard for Borda, Condorcet, Maximin and Copeland. To this end, we mainly study a variation of the above mentioned party-based election, where 
there is only one stable party, that is, the members of the instable parties can only join this stable party. This variation could be of particular interest for one participating party to determine how hard it is to manipulate or defend the outcome of the voting by persuading members of other parties to join it.

\section{Related Works}
Perek et al. \cite{DBLP:conf/atal/PerekFPR13} also considered the party-based elections, where there is a ``leading'' party with a favorable prediction. The main goal is to calculate how safe is the leading party with respect to losing its members to other parties. Hereby, Perek et al. also compute two parameters, the minimum number of members to lose to change the outcome ({\sc pes}) and the maximum number of members to lose without changing the outcome ({\sc opt}). The model by Perek et al. shares certain similarities with ours, distinguishing stable and instable parties and voters switching from instable to stable parties. Thus, our work without restriction on the number of ``instable'' parties can be considered as complementing and extending the one by Perek et al. \cite{DBLP:conf/atal/PerekFPR13}, where there is only one fixed instable party. The difference between our problems and the ones of Perek et al. \cite{DBLP:conf/atal/PerekFPR13} consists mainly in the parties whose voters may switch. Perek et al. \cite{DBLP:conf/atal/PerekFPR13} fixed one party as the winning party and all the switching of voters takes place from this party to other parties. However, in our problems the voters from any party can switch to any other party without any such restrictions. Besides this, the motivation of our work is significantly different from that of theirs. Our model puts strong emphasis on the stability of the election as a whole, and the instability could make some parties better off or worse off. This issue is within the compass of the voting rule designer. However, their work measures the stability related to a fixed party leader (the leader of the fixed party which is assumed to be known before performing the strategic behavior). Moreover, although the complexity results of our problems seem very similar to the ones achieved by Perek et al. \cite{DBLP:conf/atal/PerekFPR13}, the complexity results of one model cannot be inferred from that of the other. From the technical point of view, our reductions are based on completely different reduction strategies compared to Perek et al. \cite{DBLP:conf/atal/PerekFPR13}.

Our study has clear connection to the bribery problem \cite{DBLP:journals/jair/FaliszewskiHH09}, where voters may be bribed to change their votes in any possible way to influence the voting outcome. In contrast, we consider party-based elections, where voters can only switch from party to party and follow the party preferences, which could be more realistic in many settings.

\section{Preliminaries}
In this section, we introduce some basic notions which we use throughout this paper. More detailed definitions and results about voting theory can be found in \cite{BetzlerBCN12}. An
{\it election} is a pair $\mathcal{E=(C,V)}$, where $\mathcal{C}=\{c_1, \dots, c_m\}$ is a set of {\it candidates} and $\mathcal{V}=\{v_1, \dots, v_n \}$ is a set of {\it voters}. Each voter casts a preference over $\mathcal{C}$.
A {\it preference} is a linear order that ranks the candidates from the most preferred one to the least preferred one. For example, if $\mathcal{C}=\{a,b,c\}$ and some voter likes
{\it a} best, then {\it b}, and then {\it c}, then his or her preference is represented as $a\succ b\succ c$.
For two distinct candidates $c$ and $d$, we define $N_{\mathcal{E}}(c,d)$ as the number of voters in $\mathcal{E}$ who prefer $c$ to $d$.
We omit the index $\mathcal{E}$ if it is clear from the context. We say a candidate $c$ beats (resp. ties) another candidate $c'$ if $N(c,c')>N(c',c)$ (resp. $N(c,c')=N(c',c)$).

A {\it voting rule} is a function $R$ that given an election $\mathcal{E=(C,V)}$ returns a subset $R(\mathcal{E}) \subseteq \mathcal{C}$ of the candidates
that are said to win the election. 

In this paper, we consider the following voting rules.
An $m$-candidate {\it positional scoring rule} is defined through a non-increasing vector $\alpha= (\alpha_1, \dots, \alpha_m )$ of non-negative integers. A candidate $c \in \mathcal{C}$ is assigned $\alpha_i$
points from each voter that ranks $c$ in the $i^{th}$ position of his
preference. The score of a candidate is the sum of points he gets from all voters. The candidate(s) with the maximum score are the winner(s).
Many election rules can be considered as positional scoring rules. We study the following scoring rules (for $m$ candidates)  in this paper:
Plurality (scoring vector $(1,0,\dots, 0)$), Veto (scoring vector $(1,1,\dots, 1, 0)$), $r$-Approval (scoring vector  with $r$ ones followed by $m-r$ zeroes,
and Borda (scoring vector $(m-1, m-2, \dots, 0)$). 

A {\it Condorcet-consistent rule} always elects the Condorcet winner, if it exists. The Condorcet winner is the candidate
who beats all other candidates. Examples of Condorcet-consistent rules, that will be considered in this paper, are
Maximin and Copeland. For a candidate $c$ in an election, let $B(c)$ be the set of candidates which are beat by $c$ and let $T(c)$ be the set of candidates which tie with $c$. Then, the Copeland$^{\alpha}$ score of $c$ is $|B(c)|+\alpha\cdot |T(c)|$, for $0\leq \alpha \leq 1$. A candidate is a Copeland$^{\alpha}$ winner if it has the highest score. On the other hand, the maximin score of a candidate $c$ is given by min$_{d\in \mathcal{C}\setminus \{c\}}N_{\mathcal{E}}(c,d)$, and the winner in a maximin election is a candidate with the highest score. 


We consider the election scenario where voters are grouped into {\it parties}. That is, in addition to the set of voters $\mathcal{V}=(v_1, \dots, v_n)$
we have a partition $\mathcal{P}=(P_1, \dots, P_l)$ of voters, where
$P_1, \dots, P_l$ are called parties.
We assume that  all voters in the same party vote in the same way, called the preference of the party.
An {\it election with parties} is therefore represented as a triple $\mathcal{E=(C, V, \mathcal{P})}$.
When we say that a voter {\it switches} from its original party to another party, we mean that the respective voter casts his vote according to the preference of the destination party.
We mainly study the variation, where there is only one stable party, called one destination model. Here, the members of the only stable party cannot switch to other parties. The two problems considered in this paper are defined as follows. Note that both problems have a distinguished candidate $p$, who is the winner of the election, if all voters follow their party preference.

\begin{description}\itemsep0pt
 \item {\sc One-Destination-Min}
  \item {\bfseries Input:}  An election with parties $\mathcal{E=\{C,V,P\}}$, a positive integer $k$, a distinguished candidate $p$ and a specific voting rule.
  \item {\bfseries Question:} Is there one party $P$ such that another candidate $p'\neq p$ becomes the winner after at most $k$ voters not in $P$ switching to $P$?
 \end{description}

 \begin{description}\itemsep0pt
 \item {\sc One-Destination-Max}
  \item {\bfseries Input:}  An election with parties $\mathcal{E=\{C,V,P\}}$, a positive integer $k$, a distinguished candidate $p$ and a specific voting rule
  \item {\bfseries Question:} Is there one party $P$ such that $p$ remains the winner after at least $k$ voters not in $P$ switching to $P$?
 \end{description}

Similarly, we can define {\sc Multiple-Destination-Min} and {\sc Multiple-Destination-Max} where the number of the destination parties is not restricted.

Our hardness proofs are reduced from  the following problems:

 \begin{description}\itemsep0pt
 \item {\sc Exact Three Set Cover} (X3C)
  \item {\bfseries Input:}  A set ${X}=\{x_1, \dots, x_{m}\}$, and a collection $\mathcal{S} = \{S_1, \dots, S_n\}$ of 3-element subsets of $X$.
  \item {\bfseries Question:} Does $\mathcal{S}$ have an exact cover $S$ for $X$, i.e., a subcollection $S \subseteq \mathcal{S}$
                             such that every element of $X$ occurs in exactly one subset of $S$ ?
 \end{description}
Throughout this paper, we assume that each element $x_i$ occurs in exactly three subsets of $\mathcal{S}$. This assumption does not
change the NP-hardness of X3C \cite{Gonzalez85}.

An independent set of a graph is a subset of vertices where no edge exists between any  pair of vertices in this subset.
A vertex cover of a graph is a subset of vertices whose removal results in an independent set.

\begin{description} \itemsep0pt
\item {\sc{Independent Set}} {(\sc{IS})}
\item {\bfseries Input:} A graph $G$ and an integer $t\geq 0$.
\item {\bfseries Question:} Is there an independent set of $G$ of size at least~$t$?
\end{description}

\begin{description}\itemsep0pt
\item {{\sc Vertex Cover ({\sc VC})}}
\item {\bf{Input:}}  A graph $G$ and an integer $t\geq 0$.
\item {\bf{Question:}} Is there a vertex cover of $G$ of size at most $t$?
\end{description}

X3C, {\sc Vertex Cover} and {\sc Independent Set} are known to be NP-complete \cite{GJ79}.

\begin{table}\begin{center}
\begin{tabular}{|l|c|c|}
\hline

Election & \multicolumn{2}{c|}{The complexity of computing}\\ \cline{2-3}
systems & {\ \ \ \ \ \ \sc Min\ \ \ \ \ \ \ } & {\sc Max} \\ \hline

Plurality & P & P \\ \hline

Veto & P & P \\ \hline


$r$-Approval & P & P \\ \hline



Borda & P & NP-h \\ \hline

Condorcet & P & NP-h \\ \hline

Maximin & NP-h & NP-h \\ \hline

Copeland$^{\alpha}$ & NP-h & NP-h \\ \hline
\end{tabular}
\caption{Summary of Our Results}

\label{tab:summary}
\end{center}\end{table}

{\bf{Remarks:}} Our results hold for both unique-winner and nonunique-winner models, and also for both one-destination and multiple-destination cases. Here, for the sake of simplicity, all our proofs are for the one-destination with the unique-winner model. Other combinations can be shown by slightly modifying the proofs given here. 

\section{The Complexity of Computing {\sc MIN}}
In this section, we study the {\sc One-Destination-Min} problems. In particular, we prove that these problems
 are polynomial-time solvable under all positional scoring rules. As for the Condorcet-consistent rules, we prove that the Condorcet
rule behaves in the same way as the positional scoring rules, whereas both the Maximin voting and the Copeland${^{\alpha}}$ voting for all $0\leq \alpha\leq 1$
 lead  to NP-hardness in the {\sc One-Destination-Min} problem. Our main results are summarized in the following theorems.

\begin{theorem}
{\sc One-Destination-Min} for all positional scoring rules is polynomial-time solvable. 
\end{theorem}

\begin{proof}
We prove the theorem by proposing an algorithm which runs in polynomial time. Without loss of generality, let $\vec{\alpha}=(\alpha_1,\alpha_2,...,\alpha_m)$ be
the scoring vector where $\alpha_1\geq \alpha_2\geq,...,\geq \alpha_m$. In the first step of our algorithm we guess the replacing candidate $p'$ which will have
a score at least that of $p$ after some voters switch their parties.
Clearly, our guess involves candidates only from $\mathcal{C} \setminus \{p\}$ whose number is bounded by $m-1$. Then, for each such
guess $p'$, we need to check if it is possible to make $p'$ have a score at least that of $p$ by switching at most $k$ voters possibly from several
parties to a certain
destination party in $\mathcal{P}$. The best possible way of decreasing
the gap between the scores of $p$ and $p'$ with switching the minimum number of voters is to fix a party, whose preference achieves the maximum value of $s_{\succ}(p')-s_{\succ}(p)$, as the destination. Here, $\succ$ denotes the preference of the party and $s_{\succ}(c)$ denotes the score of the candidate $c$ from $\succ$. 
Now, we sort the party preferences of the remaining parties according to the non-increasing order of $s_{\succ}(p)-s_{\succ}(p')$.
Finally, we switch the voters of the party ordered at the first place, then the one at the second place, and so on to the destination party, until $k$ voters are switched or the score of $p'$ is at least that of $p$. If the latter case applies, we return ``yes''; otherwise, we return ``no''. The correctness and running time of the algorithm are easy to prove.

\end{proof}

\begin{theorem}
{\sc One-Destination-Min} for the Condorcet voting rule is solvable in  polynomial time.
\end{theorem}
\begin{proof}
Again, the algorithm first guesses the candidate $p'$ which beats $p$ in the final election.
Next, it fixes one party, whose party preference   prefers $p'$ to $p$, as the destination party.
Then, it switches arbitrary $k$ voters, which prefer $p$ to $p'$, to the destination party and checks the final winning status of $p$ and $p'$.
For each guessed candidate, the switch of voters and the calculation of scores
can be done in polynomial time, and with at most $m-1$ such guesses we have an overall polynomial-time algorithm.
\end{proof}

\begin{theorem}
{\sc One-Destination-Min} for the Copeland$^{\alpha}$ voting rule is NP-hard, for every $0\leq \alpha\leq 1$.
\end{theorem}

\begin{proof}
We reduce an instance $G=(V,E)$ of {\sc Vertex Cover} with $|V|=n$ to an instance $\mathcal{E=(C,V, P)}$ of {\sc One-Destination-Min}.
Clearly, {\sc Vertex Cover} remains NP-hard with $t < \frac{n-5}{2}$.
Without loss of generality, assume $n$ being even.
We further assume that  there exist two vertices $v'$ and $v''$
in $V$, which do not belong to some solution set (e.g., both $v$ and $v'$ have degree-1).
Both assumptions do not change the NP-hardness of {\sc{Vertex Cover}}.

For each edge $e_i \in E$, we create a corresponding candidate in $\mathcal{C}$. With slight abuse of terminology, we use the same notation
to denote the candidate as its corresponding edge in $E$. Let $E(v)$ be the set of candidates corresponding to the edges containing the vertex $v$
and $E^*$ be the set of all candidates corresponding to the edges in $E$.
 In addition, we have six candidates $ p, A=\{a_1,a_2\}$ and $B=\{b_1,b_2,b_3\}$.
In all the following preferences, we assume an order $(a_1 \succ a_2)$ for $A$ and an order $(b_1 \succ b_2 \succ b_3)$ for $B$.
In the following preferences, the elements of some subset $E' \subseteq E^*$ are ordered consecutively. Hereby, we use $\dots \succ E' \succ \dots$ to denote the suborder formed by $E'$ and the elements in $E'$ are assumed to be ordered in this suborder according to their indices. We first create the following two preferences:

$E(v') \succ A \succ p \succ B \succ E^* \setminus E(v')$ and

$E(v'') \succ A \succ p \succ B \succ E^* \setminus E(v'')$.

In addition, for every other vertex $v\in V\setminus \{v',v''\}$, we create a preference defined as follows:

$E(v)\succ p\succ B\succ A\succ E^* \setminus E(v)$

Each of the above preferences represents a party. We denote by $P_v$ the party corresponding to the vertex $v$. Furthermore, we have a party $P$ containing one voter with the following preference:

$B \succ E^* \succ p \succ A$.

Finally, we have $n-2$ voters, out of which the first $\frac{n-2}{2}$ voters  form a party denoted by $P_1$ with the preference:

$a_1\succ p \succ a_2 \succ E^* \succ B$.

The other $\frac{n-2}{2}$ voters form a party denoted by $P_2$ with the following preference:

$  E^* \succ B \succ a_1\succ p \succ a_2$.

Finally, set $k=t$. Before discussing the correctness, consider the score of each candidate first.
%
%
For a candidate $c$, let $s(c)$ be the Copeland${^{\alpha}}$ score of $c$. Then we have 

$s(p)=|E|+4$,

$s(a_1)=|E|+2$, 

$s(a_2)=|E|$,

$s(b_i)=5-i$, where $i=1,2,3$, and

$s(e_i)=|E|-i+3$.

It is clear that $p$ is the current winner. Some useful observations are as follows:

\begin{claim}\label{claimscoreproperty} The following claims hold:

(1) $s(p)$ cannot be decreased by switching at most $k$ voters.

(2) $s(a_i)$, $s(e_i)$ cannot be increased by switching at most $ k$ voters.

(3) Switching of at most $k$ voters can increase $s(b_1)$ to at most $|E|+4$.
\end{claim}

\begin{proof}
(1) Since for every $e_i\in E^*$ we have $N(p,e_i)-N(e_i,p)=n-5$, and for every $b_i\in B$, we have $N(p,b_i)-N(b_i,p)=n-1$, with the assumption that $t<\frac{n-5}{2}$,
switching arbitrary $k=t$ voters can never decrease the score of $p$.

(2) Similar to (1).

(3) Observe that $N(e_i,b_1)-N(b_1,e_i)=1$ for every $e_i\in E^*$. Therefore, $b_1$ has  the
potential to beat every $e_i$: just switch one voter of a party with $e_i \succ b_1$ to a party with $b_1 \succ e_i$.
Therefore, $b_1$ has the potential to have a score of $|E|+4$.
However, since $N(p,b_1)-N(b_1,p)=n-1$, $b_1$ has no chance to increase its score furthermore. \end{proof}.

Now we prove that $G$ has a vertex cover of size at most $t$, if and only if {\sc One-Destination-Min} on $\mathcal{E=\{C,V,P\}}$ has a ``yes'' answer.

$(\Rightarrow:)$ Let $G$ have a vertex cover $C$ of size $t$. Consider the election after all voters
corresponding to $C$ switch to the party $P$. Since $C$ is a vertex cover, for every edge $e_i$, there is at least
one vertex $v \in C$ with $e_i \in E(v)$. Therefore, for every edge $e_i$, at least one voter of a party with preference  $e_i \succ b_1$ is switched to the
party $P$, where $b_1 \succ e_i$. Due to the analysis of the third claim in Claim \ref{claimscoreproperty}, $b_1$ beats every $e_i$ and thus $p$ is not the unique winner anymore.

$(\Leftarrow:)$ Assume that  {\sc One-Destination-Min} on $\mathcal{E=\{C,V,P\}}$ has a ``yes'' answer and $C$ is the set of voters which switch to the destination party.
Due to Claim \ref{claimscoreproperty}, the only candidate which could have a score at least that of $p$ is $b_1$. This can only  happen if $b_1$
beats every $e_i \in E^*$. Based on this claim, we observe that the parties $P_1$ and $P_2$ cannot be the destination party.
Among the remaining parties, it is obvious that $P$ is the best possible destination party since $b_1$ beats every $e_i$
in this party (in other words, if there is a solution in which the destination party is not $P$, we can always construct another
solution with $P$ being the destination party). Now  consider which parties could be the instable parties. We claim the following.

\begin{claim}
The voters of $P_1$ and $P_2$ cannot be switched.
\end{claim}

To verify the above claim, observe that switching one arbitrary voter from $P_1\cup P_2$ to $P$ would make $b_1$
reach its highest possible score $|E|+4$. However, this also increases the score of $p$ by one (from beating $a_1$);
thus $p$ remains the winner.

Now we show that the vertices corresponding to the voters in $C$ must be a vertex cover in $G$.
Due to the above analysis, $b_1$ has a score at least that of $p$, only if $b_1$ beats every $e_i \in E^*$.
Therefore, for every edge $e_i$, there must be at least one voter in $C$
corresponding to a vertex $v$ with $e_i \in E(v)$,
implying the vertices corresponding to $C$ form a vertex cover of $G$.
\end{proof}

\begin{theorem}
{\sc One-Destination-Min} for Maximin is NP-hard.
\end{theorem}

Due to space limitation, the proof is deferred to the Appendix.

\section{The Complexity of Computing MAX}

In this section, we prove the polynomial-time solvability of {\sc One-Destination-Max} for Plurality, $r$-Approval and Veto rules and present NP-hardness results of the
same problem for Borda, Condorcet, Maximin and Copeland rules.

%

\begin{theorem}
{\sc One-Destination-Max} for Plurality, $r$-Approval with constant $r$ and Veto voting rules are polynomial-time solvable. 
\end{theorem}
\begin{proof}
We consider here only $r$-Approval. The cases with Plurality and Veto can be handled similarly. Our polynomial-time algorithm first guesses the destination party among all the parties in the given instance $\mathcal{E=(C,V,P)}$.
Among the total of $|\mathcal{V}|$ such guesses, we discard those with preferences which do not approve $p$.
For each remaining guessed destination parties we do the following. 

Let $C \subseteq \mathcal{C}$ be the set of $r$ candidates which are approved by the preference of the destination party.
Since $p$ is the unique winner, for each candidate $c\in {C} \setminus \{p\}$, there must exist at least one voter disapproving $c$ in the original election.
To maintain $p$ as the unique winner,
for each candidate $c\in C \setminus \{p\}$, there must be at least one voter disapproving $c$ in the original election and this voter cannot
be switched to the destination party. Therefore, we need to find out a minimum set of voters together disapproving ${C}\setminus \{p\}$. Since $|C|\leq r$ and $r$ is a
constant, this set can be found in polynomial time. 
\end{proof}

\begin{theorem}
{\sc{One-Destination-Max}} for the Borda rule  is NP-hard.
\end{theorem}

\begin{proof}
We reduce from X3C.
Let $(X=\{x_1,x_2,...,x_m\},\mathcal{S}=\{S_1,S_2,...,S_n\})$ be an instance of X3C.

We have in total $m+6$ candidates. More precisely, for each $x_i \in X$, we have a corresponding candidate. For simplicity, we will use the same notation $x_i$ to denote the corresponding candidate.
In addition, we have six other candidates $p,d_1,d_2,d_3,y,z$. We create party preferences as follows.

Let $\overrightarrow{X}$ be the order of the elements in $X$ according to the increasing order of their indices.  For each subset $\{x_i,x_j,x_k\} \in \mathcal{S}$ with $i < j < k $, we create one preference $z\succ \overrightarrow{X}[x_i\rightarrow d_1, x_j\rightarrow d_2, x_k\rightarrow d_3]\succ p\succ x_i\succ x_j\succ x_k\succ y$.
Here $\overrightarrow{X}[x_i\rightarrow d_1, x_j\rightarrow d_2, x_k\rightarrow d_3]$ is the linear order obtained from $\overrightarrow{X}$ with
replacing $x_i, x_j, x_k$ by $d_1, d_2, d_3$, respectively. The corresponding party is denoted by $P_{(i,j,k)}$ and has only one voter.

Next, we create one party $P'$ with $n$ voters  and the preference:

 $y\succ p\succ \overleftarrow{X}\succ z\succ d_1\succ d_2\succ d_3$, where $\overleftarrow{X}$ is the reverse order of $\overrightarrow{X}$.
Additionally, we create a party $P$ with one voter and the following preference:

$\overrightarrow{X}\succ p\succ d_1\succ d_2\succ d_3\succ y\succ z$.

It is clear that $p$ is the current winner. More precisely, $s(p)-s(z)=5, s(p)-s(y)=3n+4$ and $s(p)-s(x_i)> 0$. Here, $s(c)$ denotes the Borda score of $c$. We claim that an exact 3-set cover exists if and only if $n-m/3$ voters can switch their parties to a party such that  $p$ is still the winner.

Suppose that there is an exact 3-set cover $S$ for $(X,\mathcal{S})$. Then leave all the parties corresponding to $S$ and the party $P'$ unchanged, and switch all the  other voters into the party $P$. It is easy to check that $p$ is still the winner.

For the reverse direction, we first claim that only $P$ can be the destination party if the constructed instance is a true-instance. $P'$ cannot be the destination party, since otherwise, $y$ would become the winner.  $P_{(i,j,k)}$ cannot be the destination party, since otherwise, either $z$ or some $x_i$ would become the winner. This completes the proof of claim. We further claim that the party $P'$ cannot be instable. This is true, since otherwise,
some $x_i$ would have a higher score than that of $p$. Now suppose that there is no exact 3-set cover. Then there must be an $x_i$ such that after switching $n-m/3$ voters from $\cup_{i,j,k}P_{(i,j,k)}$ to the party $P$, all the remaining voters in the parties $P_{(i,j,k)}$ prefer $x_i$ to $p$. This results in $x_i$ having a greater score than that of $p$, and thus $p$ cannot be the unique winner anymore.
\end{proof}

\begin{theorem}
  {\sc One-Destination-Max} for the Condorcet voting rule is NP-hard.
\end{theorem}
\begin{proof} We give a reduction from an X3C-instance $(X,\mathcal{S})$ to an instance $\mathcal{E=\{C,V,P\}}$ of {\sc One-Destination-Max}.
  The candidate set is $\mathcal{C}= X\cup A\cup B\cup C\cup D\cup \{p\} $, where $A=\{a_1,a_2,a_3,a_4\}~, B=\{b_1,b_2,b_3,b_4\},
  C=\{c_1,c_2,...,c_n\}$ and $D=\{d_1,d_2,...,d_n\}$. The set $X$ contains a candidate for each element in $X$. Hence, there are totally $2n+m+9$ candidates.
  For a set $U$ with
  elements $u_1,u_2,...,u_t$, we denote by $\overrightarrow{U}$ the order $(u_1 \succ u_2 \succ \dots \succ u_t)$  and  by $\overleftarrow{U}$
  the reverse order of $\overrightarrow{U}$. Moreover, for two elements $u_i,u_j$ with $i<j$, we use $\overrightarrow{U}[u_i,u_j]$
 to denote the suborder $(u_i \succ u_{i+1} \succ \dots \succ u_j)$.
  We assume a fixed order $(S_1 \succ S_2 \succ \dots \succ S_n)$ for the subsets in $\mathcal{S}$. We create two preferences for each $S_t \in \mathcal{S}$ with
   $S_t=\{x_i,x_j,x_k\}$:

1. $\overrightarrow{D}[d_{t+1}, d_n]\succ \overrightarrow{C}[c_1,c_t]\succ \overrightarrow{A}
\succ x_i\succ x_j\succ x_k \succ p \succ \overrightarrow{X}\setminus S_t\succ \overleftarrow{B}\succ
\overrightarrow{C}[c_{t+1},c_n] \succ \overrightarrow{D}[d_1,d_t]$. The corresponding party denoted by $P_t$, has only one voter.

2. $\overrightarrow{D}[d_1,d_t]\succ \overrightarrow{C}[c_{t+1},c_n]\succ \overrightarrow{B} \succ
\overleftarrow{X}\setminus S_t \succ p\succ x_k\succ x_j\succ x_i\succ \overleftarrow{A} \succ
\overrightarrow{C}[c_{1},c_t]\succ \overrightarrow{D}[d_{t+1},d_n]$. The corresponding party, denoted by $P_t'$, has only one voter.

Next, we construct  additional preferences as follows, each representing a party of its own. Each of the parties has only one voter.

$a_1 \succ b_1 \succ p \succ \overrightarrow{X} \succ \overrightarrow{A}\setminus \{a_1\} \succ
\overrightarrow{B}\setminus \{b_1\} \succ \overrightarrow{C}\succ \overrightarrow{D}$; the corresponding party is denoted by $\bar{P}_1$

$a_2 \succ b_2 \succ p \succ \overleftarrow{X} \succ \overrightarrow{A}\setminus \{a_2\} \succ
\overrightarrow{B}\setminus \{b_2\} \succ \overleftarrow{C}\succ \overleftarrow{D}$; the corresponding  party is denoted by $\bar{P}_2$

$a_3 \succ b_3 \succ p \succ \overrightarrow{X} \succ \overrightarrow{A}\setminus \{a_3\} \succ
\overrightarrow{B}\setminus \{b_3\} \succ \overrightarrow{C}\succ \overrightarrow{D}$; the corresponding party  is denoted by $\bar{P}_3$.

$a_4 \succ b_4 \succ p \succ \overleftarrow{X} \succ \overrightarrow{A}\setminus \{a_4\} \succ
\overrightarrow{B}\setminus \{b_4\} \succ \overleftarrow{C}\succ \overleftarrow{D}$; the corresponding party is denoted by $\bar{P}_4$.

$\overrightarrow{X} \succ p \succ \overrightarrow{A} \cup \overrightarrow{B} \succ
\overrightarrow{C} \succ \overrightarrow{D}$; the corresponding party is denoted by $P$.

In total, we have $2n+5$ voters. Now we arrive at the correctness proof of the reduction.









$(\Rightarrow:)$ Clearly, $p$ beats every other candidate  and thus is the current winner. Suppose that $(X,\mathcal{S})$ has an exact 3-set cover $S$. We claim that after switching all the voters in the parties $P_t'$, which correspond to the subsets in  $S$, to the party $P$, $p$ will still be the winner. Observe that $p$ beats every candidate in $C\cup D\cup A\cup B$ in the final election. Since $S$ is an exact 3-set cover, for each $x\in X$ there is exactly one party $P_t'$ in the solution preferring $p$ to $x$.
 Even though  the party $P$ prefers $x$ to $p$, $p$ still beats $x$ by $n+3$. The claim follows.

$(\Leftarrow:)$ Suppose that we switch a set $S'$ of $m/3$ voters to a particular party in $\mathcal{P}$ such that $p$ remains the Condorcet winner.
We claim the following:

\begin{claim}
No $\bar{P}_i$, where $i=1,2,3,4$, can be the destination party.
\end{claim}
\begin{proof}
Due to the symmetry, we only need to give the proof for the party $\bar{P}_1$. All other cases are similar.
Observe that all parties other than $\bar{P}_1$ prefer $p$ to  either $a_1$ or $b_1$, or both.
Since there are $n+1$ parties preferring $a$ to $p$ in the original election, switching any arbitrary five voters to the party $\bar{P}_1$ will make $a_1$ or $b_1$ beat $p$, contradicting with the fact that $p$ is the Condorcet winner in the final election.
\end{proof}

\begin{claim}
None of $P_t$ and $P_t'$ can be the destination party, for all $t\in \{1,2,...,n\}$.
\end{claim}
\begin{proof}
Due to the symmetry, we only need to give the proof for $P_t$ for a certain $t$. Suppose that this is not true and we have switched a set $S'$ of $m/3$ voters to the party $P_t$, where $S_t=\{x_i,x_j,x_k\}$, without changing the winning candidate. There is at most one voter in $S'$ which is not from the parties $\bigcup_{z\in \{1,...,n\}} P_z$, since otherwise, some $a_i$ would beat $p$, contradicting that $p$ is the Condorcet winner. Therefore, at least $m/3-1$ voters of $S'$ are from $\bigcup_{z\in \{1,...,n\}} P_z$. Moreover, at most two voters of $S'$ are from $\bigcup_{z>t}P_z$, since otherwise, $d_t$ would beat $p$. Symmetrically, at most two voters of $S$ are from $\bigcup_{z<t}P_z$ (otherwise, $c_t$ would beat $p$), implying that $|S'|\leq 5$, a contradiction.
\end{proof}

Due to the above two claims, the only possible destination party is $P$. We further claim the following facts.

\begin{claim}
None of the parties $\bar{P}_i$ can be instable, where $i=1,2,3,4$.
\end{claim}
\begin{proof}
Observe that $p$ beats every $x_i$ by $n+4$. Therefore, if we switch some voter in $\bar{P}_i$ to the party $P$, then no other voter can be switched to $P$, since every voter not in the party $P$ prefers $p$ to some $x_i$. The claim follows.
\end{proof}

\begin{claim}
None of the parties $P_t$ can be instable for $t\in \{1,2,...,n\}$.
\end{claim}
\begin{proof}
Suppose that we switch some voter in a certain party $P_t$ to $P$, where $S_t=\{x_i,x_j,x_k\}$. Due to Claim 5, no voter is switched  from $\bigcup_z \bar{P}_z$ to $P$. 
Besides, at most one voter in $\bigcup_z P_{z}'\cup P_z\setminus \{P_t\}$ can be switched to $P$, since otherwise, some $x_i$ would beat $p$, contradicting that $|S'|=m/3$.
\end{proof}

According to the above claims, the instable parties can only be from  $S'\subseteq \bigcup_{z\in \{1,2,...,n\}} P_z'$. Since $p$ beats every $x\in X$ by $n+1$, at most one voter in $S'$ prefer $p$ to $p$, implying that the subsets corresponding to $S'$, that is $S=\{\{x_i,x_j,x_k\}\mid \exists{t}, P_t'\in S'~{and}~S_t=\{x_i,x_j,x_k\}\}$, form an exact 3-set cover.
\end{proof}

\begin{theorem}\label{maxmaximin}
{\sc One-Destination-Max} for the Maximin rule is NP-hard.
\end{theorem}

Next we show that {\sc One-Destination-Max} remains hard for the Copeland voting rule.

\begin{theorem}\label{thm:copemax}
{\sc One-Destination-Max} for Copeland$^{\alpha}$ is NP-hard, for every $0\leq \alpha\leq 1$.
\end{theorem}

Due to space limitation, the proofs for the above two theorems are deferred to the Appendix.

\section{Conclusion}
We examined the election systems with parties. Here, parties can be partitioned into stable and instable parties. Since members of instable switch to stable parties, the outcome of the election could be diverse to the prediction made based on the preferences of the parties. We introduced two parameters {\sc{Min}} and {\sc{Max}} to measure the credibility of the prediction of such elections and present a comprehensive study of the complexity for computing {\sc{Min}} and {\sc{Max}} under the most common positional scoring rules {Plurality, $r$-Approval, Borda, Veto} and three Condorcet-consistent rules (Condorcet, Maximin, Copeland).

An avenue for possible future research could be to investigate other variants of the model studied here. For instance, practical applications indicate that the members of one party can switch only to the parties, which have similar preferences as their own. More formally, a member of the party $P$ can switch to another party $P'$, only if the "distance" between the preferences of $P$ and $P'$ is bounded. Here, the distance measure could be swap-distance, Kendall-Tau-distance, etc. It could be also of practical interest to assume that the number of members leaving an instable party is bounded. 

 \bibliographystyle{aaai}	
 \bibliography{myrefs}

\input{appendixx}

\end{document}

%% file: appendixx.tex
\section*{Appendix}

{\bf{Proof of Theorem 4}}
\begin{proof}
We give a reduction from an X3C instance $(X,\mathcal{S})$ to an instance $\mathcal{E=\{C,V,P\}}$ of {\sc One-Destination-Min}.

For each $x\in X$, we create a candidate. For convenience, we still use $x$ to denote the candidate. In addition, we have four candidates $p, z, \alpha$ and $\beta$. For each set ${S}_t \in \mathcal{S}$ with $S_{t}=\{x_i, x_j, x_k\}$ and $i < j < k$, create the following preference:

$x_i \succ x_j \succ x_k \succ z \succ p \succ X \setminus \{x_i, x_j, x_k\} \succ \beta \succ \alpha$.

This preference represents a party with only one voter.


Next, create $n-\frac{m}{3}$ new voters; the half of them forms one party with the following preference:

 $\beta \succ \alpha \succ p \succ \overrightarrow{X} \succ z$.

The other half has the preference:

 $\alpha \succ p \succ \overrightarrow{X} \succ z \succ \beta$.

Let $B_1,B_2,...,B_{m/3}$ be subsets of $X$, where $B_i=\{x_{3i-2},x_{3i-1},x_{3i}\}$.
For each $B_i$, create two voters forming two parties with the following two preferences, respectively:

$\beta \succ \alpha \succ p \succ {X}\setminus B_i \succ z \succ B_i$.

$\alpha \succ p \succ {X}\setminus B_i \succ z \succ B_i\succ \beta$.

Finally, we create a party with one voter and the preference: $z \succ \overrightarrow{X} \succ p \succ \alpha \succ \beta$. This party is denoted by $P$, which we later prove to be the destination party.
Overall, $|\mathcal{C}|=m+4$ and $|\mathcal{V}|=2n+\frac{m}{3}+1$.
Next, we calculate the score of each candidate.
In the following, $s(c)$ denotes the maximin score of the candidate $c$, and $\min (c)$ is the set of candidates $c'$, which reach the minimum value of $N(c,c')$:

$s(\beta)=(n + \frac{m}{3}) / 2 $ and $\min{(\beta)}=\{p, z, x_i \}$,

$s(\alpha)= (n + \frac{m}{3}) / 2+1 $ and $\min{(\alpha)}=\{\beta\}$,

$s(p)=n+1$ and $\min{(p)}=\{\alpha\}$,

$s(z)=n$ and $\min{(z)}=\{x_i\}$,

$s(x_i)=4$ and $\min{(x_i)}=\{p\}$.

Clearly, $p$ is the current winner. Suppose that there is an exact 3-set cover $S$. If all the voters
corresponding to $S$ switch to the party $P$, $z$ would beat every $x_i$ by $n+1$; and thus $p$ is not the unique winner anymore. It remains to show the other direction. Observe that the only candidate, which
could have a score at least that of $p$ is $z$. Therefore, if $n/3$ voters switch to some party to make $p$ not the unique winner, the destination party
can only be the party $P$. The instable parties can only be the ones corresponding to $\mathcal{S}$, since with other parties being instable, the score of $p$ would increase.
However, the set corresponding to the voters which are switched to the destination party must be an exact 3-set cover, since otherwise, there would exist an $x_i$ such that at
most $n$ voters prefer $z$ to $x_i$, resulting in $p$ still being the winner.
\end{proof}
\bigskip

\noindent{\bf{Proof of Theorem 8}}
\begin{proof}
We give a reduction from {\sc Independent set}. Given an instance
$I=(G,t)$ of {\sc{Independent Set}}, where $G=(V,E)$, $E=\{e_1, \dots, e_m\}$, $n=|V|$, we create an instance $\mathcal{E=(C,V,P)}$ of  {\sc One-Destination-Max} as follows.

The candidate set is $\mathcal{C}=\{a, b, p, e_1, \dots, e_m\}$.
 Let $\overrightarrow{E}=(e_1 \succ e_2 \succ \dots \succ e_m)$ be an order of $E$.
We first create a set $Z$ of $n$ parties corresponding to the vertices in the graph.
More specifically, for each vertex $v \in V$, we create a preference defined as follows:

$a \succ \overrightarrow{E \setminus E(v)} \succ p \succ \overrightarrow{E(v)}  \succ b $.

Each preference represents a party with one voter, denoted by $P_v$.
Then, we create a party $P'$ containing $n$ voters with the preference:

$b \succ p \succ \overleftarrow{E} \succ a$.

Finally, we create a party $P$ containing only one voter with the preference:

$\overrightarrow{E} \succ p \succ b \succ a$.

\begin{table}\begin{center}
\begin{tabular}{|c|c|c|c|c|c|} \hline
 & $p$ & $b$ & $e_j(i > j)$ & $e_j (i < j)$ & $a$ \\ \hline

$p$ & - & $n+1$ &$n+2$ & $n+2$ & $n+1$\\ \hline

$b$ & $n$ &- & $n$& $n$ & $n+1$   \\ \hline

$e_i$ & $n-1$ & $n+1$ & $\geq n$ & $\geq n-1$ & $n+1$ \\ \hline

$a$ & $n$ &  $n$ &  $n$ &n &  - \\ \hline
\end{tabular}
\caption{Comparisons between candidates in proof of Theorem 8,  where the entry in the row of $x $ and the column of $y$ denotes the number of voters who prefer $x$ to $y$.}
\label{tablemaximinonedestinationmax}\end{center}
\end{table}

Finally, set $k=t$. Now we prove the correctness of the reduction. We refer to Table \ref{tablemaximinonedestinationmax} for the comparison of scores of candidates.
It is clear from the table, that $p$ is the current unique winner as $p$ is preferred to  all other candidates by at least $n+1$ voters.

$(\Rightarrow:)$ Assume that $I$ is a true-instance and $S$ is an independent set of $G$ of size $t$.
 Consider the election after $k=t$ voter corresponding to $S$ switch to the party $P$. Let $V_s$ be the set of these $k=t$ voters.
It is clear that $p$ beats $a$ and $b$ by $n+1+k$ and $n+1$, respectively.
Since $S$ is an independent set, for each edge $e_i$, there is at most one voter in $V_s$ which prefers $e_i$ to $p$.
 Hence, $p$ beats every edge candidate $e_i$ by at least $n+1$,
implying the maximin score of $p$ is $n+1$. Moreover, the scores of $a$ and $b$ do not increase.
It remains to show that score of every $e_i$ still remains less than that of $p$.
 To check this, consider the comparison of the scores of $p$ and $e_i$.
Since $S$ is an independent set, the same reason discussed above implies that every $e_i$ beats $p$ by at most $n$.
 Thus, the maximin score of $e_i$ cannot be greater than $n$, implying that $p$ still remains the unique winner.

$(\Leftarrow:)$ Assume that it is possible to switch a set $S'$ of $k$ voters in $\mathcal{V}$ from their original parties to the destination party such that $p$ still remains the winner in the overall
election.

We first claim that $P'$ cannot be the destination party. It is easy to see that if $P'$ is the destination party, $b$ will become the new winner replacing $p$. We then distinguish the following cases:

{\bf{Case 1. {\it \bfseries P} is the destination.}} In this case, we can assume that at most one voter in $S'$ is from $P'$, since otherwise, $e_1$ would replace $p$ as a winner.
Assume now that there is exactly one voter of $S'$ which belongs to the party $P'$. Clearly, all other voters of $S'$ come from the parties in the set $Z$.
Since the party $P$ prefers every $e_i$ to $p$, and the parties in the set $Z$ prefer $p$ to some edge candidates, the score of $p$ will be at most $n$.
However, $e_1$ has a score at least $n$, contradicting that $p$ is the unique winner. Based on the above fact, it is safe to  assume that all the votes in $S'$ belong to the parties in the
set $Z$. We claim now that the vertices corresponding to the voters of $S'$ form an independent set. If this is not true, there must be some edge candidates, each of which is preferred to $p$ by two voters of $S'$. Let $e_i$ be such an edge candidate with maximum index $i$. Consider the election after all the voters in $S'$ are switched to the party $P$. It is clear that $e_i$ beats $p, a, b$ by at least $n+1$, and $e_i$ beats $e_j$ for all $j<i$ by at least $n$. Now consider the comparison between some $e_j$ with $j>i$ with $e_i$. Since the voters in the set $Z$ which  prefer $e_j$ to $e_i$, are switched to the party $P'$ where the voters prefer $e_i$ to $e_j$, all the voters in the set $Z$ and all the voters in
 the party $P$ prefer $e_i$ to $e_j$, implying that $e_i$ beats $e_j$ by $n+1$. Thus, we conclude that $e_i$ has a final score of $n$. However, since $p$ beats $e_i$ by $n$ in
the final election, $p$ is not the unique winner anymore.

{\bf{Case 2.}} Now we consider the case that the destination party is some $P_v$. Again, we claim that the vertices corresponding to $S'$ form an independent set,
 if $p$ is still the unique winner. For the sake of contradiction, assume this is not true. Then there must be an edge $e_i=(u,w)$ with minimum index $i$, such that two voters in $S'$ are switched to the
party $P_v$. Note that  $e_i$ cannot be adjacent to $v$. Thus, $p$ beats $e_i$ by $n$, implying that the score of $n$ is at most $n$.
Now consider the score of $e_i$. It is easy to verify that $e_i$ beats $a,b,p$ by $n+1$, and beats $e_j$ for all $j>i$ by $n$. Moreover,
 since the only two voters in the set $Z$ which prefer $e_j$ to $e_i$ are switched to the party $P_v$, which prefers $e_i$  to $e_j$, the score of $e_i$
is at least $n$. Therefore, $p$ no longer remains the unique winner, contradicting the assumption.

%
%
%

\end{proof}
\bigskip

\noindent{\bf{Proof of Theorem 9}}
\begin{proof}
We show the NP-hardness by a reduction
from  {\sc IS}. 
Given an instance $I=(G=(V,E),t)$ of {\sc IS} where $E=\{e_1', \dots, e_m'\}$, and $ n=|V|$,
we construct the instance $I'=(\mathcal{C},\mathcal{V}, \mathcal{P})$ of {\sc One-Destination-Max} as follows:
Our candidate set is $\mathcal{C}=A \cup  B \cup  C \cup  \{p\} \cup  E^*$ where $A=\{a_1, \dots ,  a_m\}$, $B=\{b_1, \dots, b_m\}$,  $C=\{c_1, \dots, c_m\}$ and $E^*$ contains a candidate for each edge in $E$.
We construct the following set of parties. Here, the elements in $A$, $B$, $C$ and $E^*$ are ordered according
to the indices of the elements.

\noindent (1) For each $v \in V$  create a party $P_v$ containing one voter with the preference $A   \succ  E^* \setminus E(v) \succ C \succ p \succ E(v)  \succ B$ .
Here $E(v)$ denotes the set of the edges incident to the vertex $v$. Let $Z$ denote the set of the voters in these parties.\\
(2) We have one  party $P'$ containing $n$  voters with the preference $B \succ C  \succ p \succ E^*  \succ A  $

\noindent and one party $P$ containing one voter with the preference $ E^* \succ p \succ A \succ B \succ C  $.

\begin{table}\begin{center}
\begin{tabular}{|c|c|c|c|c|c|} \hline
 & $p$ & $a_i$ & $b_i$ & $e_i$ & $c_i$ \\ \hline

$p$ & - & $n+1$ &$n+1$ & $n+2$ & $1$\\ \hline

$a_i$ & $n$ &- & $n+1$& $n$ & $n+1$   \\ \hline

$b_i$ & $n$ & $n$ & - & $ n$ & $n+1$ \\ \hline

$e_i$ & $n-1$ &  $n+1$ &  $n+1$ & -  &  $n-1$ \\ \hline

$c_i$ & $2n$ & $n$ & $n$ & $n+2$ & - \\ \hline

\end{tabular}\end{center}
\caption{Comparison between candidates in the Proof of Theorem 8, where the entry in the row of $x $ and the column of $y$ denotes the number of voters who prefer $x$ to $y$.}
\label{tab:copemax}
\end{table}

Observe that we have $2n+1$ voters; thus there is no  tie. The initial scores of the candidates,
which follow directly from Table \ref{tab:copemax}, are as follows:

$s(p) = |A| + |B| + |E^*|$

$s(a_i) = (|A| -1) + |B| + |C|  $

$s(b_i) = (|B| -1) + |C|  $

$s(e_i) \leq |A| + |B| + |E^*|-1 $

$s(c_i) \leq 1 + |E^*|+ |C|-1 $

We are ready to prove the correctness.

$(\Rightarrow:)$ Let $S$ be an independent set of size $k$ in $G$. We switch all the voters corresponding to the vertices in $S$ from $Z$ to
party $P$. Since $S$ is an independent set, for every edge candidate $e_i$, there is at most one voter in $S$ preferring $p$ to $e_i$. Thus, even
after the switching of these voters, $p$ still beats every candidate $e_i \in E^*$
by at least $n+1$ voters. Thus, the score of $p$ remains unchanged.
The only  candidates, whose score may increase after the switching of voters, are  the candidates $e_i \in E^*$.
A candidate $e_i$
can have a score at least that of $p$ only if $e_i$ beats $p$ or some candidate $c_i$. However, this is impossible, since $S$ is an independent set and $e_i$ beats $p$ and
every $c_i$ by $n-2$ in the original election. Thus $p$ still remains the unique winner.

$(\Leftarrow:)$ Suppose it  is possible to  switch a  set $S'$ of $k$ voters to a destination party, such that $p$  still remains the winner. First observe that $P'$ cannot be
the destination party, since otherwise, $b_1$ would replace $p$ as the  winner. We distinguish the following two cases:

{\bf{Case 1.}}  $P$ is the destination party. Observe that irrespective of the composition of $S'$, $p$ still beats all the candidates in $A \cup B$ but
none in $C$. Moreover, no candidate in  $A \cup B \cup C$ can increase its score. Since $p$ is the unique winner in the final election, no voters in $S'$
come from the party $P'$, since otherwise, some $e_i$ would beat $p$ and thus prevent $p$ from being the unique winner. Therefore, all voters of $S'$ must be
from  $Z$. More specifically, the vertices corresponding to $S'$ form an independent set, since otherwise, some edge candidate would replace $p$ as the winner.

{\bf{Case 2.}} Some party $P_v$ is the destination party. In this case, no voter of $S'$ is from  $P \cup P'$, since otherwise, since $a_1$ would
prevent $p$ from becoming the winner. Thus, all the votes of $S'$ are from the set $Z$. We claim that the vertices corresponding to $S'$ form an
independent set. For contradiction, assume that this is not true. Then, there must be an edge $e_i$ for which there are two voters in $S'$ preferring $p$ to $e_i$.
Note that $e_i \notin E(v)$. Therefore, $p$ cannot beat $e_i$, leading to that $p$'s score is one less than that of $e_i$ in the original election. Hence, $a_1$
would prevent $p$ from becoming the unique winner, a contradiction.

\end{proof}

%% file: main_AAAI14.bbl
\begin{thebibliography}{}

\bibitem[\protect\citeauthoryear{Bartholdi, Tovey, and
  Trick}{1989}]{BARTHOLDI89}
Bartholdi, J.~J.; Tovey, C.; and Trick, M.
\newblock 1989.
\newblock The computational difficulty of manipulating an election.
\newblock {\em Social Choice and Welfare} 6(3):227--241.

\bibitem[\protect\citeauthoryear{Betzler \bgroup et al\mbox.\egroup
  }{2012}]{BetzlerBCN12}
Betzler, N.; Bredereck, R.; Chen, J.; and Niedermeier, R.
\newblock 2012.
\newblock Studies in computational aspects of voting - a parameterized
  complexity perspective.
\newblock In {\em The Multivariate Algorithmic Revolution and Beyond},
  318--363.

\bibitem[\protect\citeauthoryear{Chevaleyre \bgroup et al\mbox.\egroup
  }{2007}]{DBLP:conf/sofsem/ChevaleyreELM07}
Chevaleyre, Y.; Endriss, U.; Lang, J.; and Maudet, N.
\newblock 2007.
\newblock A short introduction to computational social choice.
\newblock In {\em SOFSEM (1)},  51--69.

\bibitem[\protect\citeauthoryear{Duggan and Schwartz}{2000}]{bbb}
Duggan, J., and Schwartz, T.
\newblock 2000.
\newblock Strategic manipulability without resoluteness or shared beliefs:
  Gibbard-satterthwaite generalized.
\newblock {\em Social Choice and Welfare} 17(1):85--93.

\bibitem[\protect\citeauthoryear{Faliszewski, Hemaspaandra, and
  Hemaspaandra}{2009}]{DBLP:journals/jair/FaliszewskiHH09}
Faliszewski, P.; Hemaspaandra, E.; and Hemaspaandra, L.~A.
\newblock 2009.
\newblock How hard is bribery in elections?
\newblock {\em J. Artif. Intell. Res. (JAIR)} 35:485--532.

\bibitem[\protect\citeauthoryear{Faliszewski, Hemaspaandra, and
  Hemaspaandra}{2013}]{DBLP:conf/atal/FaliszewskiHH13}
Faliszewski, P.; Hemaspaandra, E.; and Hemaspaandra, L.~A.
\newblock 2013.
\newblock Weighted electoral control.
\newblock In {\em AAMAS},  367--374.

\bibitem[\protect\citeauthoryear{Garey and Johnson}{1979}]{GJ79}
Garey, M.~R., and Johnson, D.~S.
\newblock 1979.
\newblock {\em Computers and Intractability: A Guide to the Theory of
  NP-Completeness}.
\newblock New York, NY, USA: W. H. Freeman \& Co.

\bibitem[\protect\citeauthoryear{Gibbard}{1973}]{Gibbard73}
Gibbard, A.
\newblock 1973.
\newblock Manipulation of voting schemes: A general result.
\newblock {\em Econometrica} 41(4):587--601.

\bibitem[\protect\citeauthoryear{Gonzalez}{1985}]{Gonzalez85}
Gonzalez, T.~F.
\newblock 1985.
\newblock Clustering to minimize the maximum intercluster distance.
\newblock {\em Theor. Comput. Sci.} 38:293--306.

\bibitem[\protect\citeauthoryear{Parkes and
  Xia}{2012}]{DBLP:conf/aaai/ParkesX12}
Parkes, D.~C., and Xia, L.
\newblock 2012.
\newblock A complexity-of-strategic-behavior comparison between schulze's rule
  and ranked pairs.
\newblock In {\em AAAI}.

\bibitem[\protect\citeauthoryear{Perek \bgroup et al\mbox.\egroup
  }{2013}]{DBLP:conf/atal/PerekFPR13}
Perek, T.; Faliszewski, P.; Pini, M.~S.; and Rossi, F.
\newblock 2013.
\newblock The complexity of losing voters.
\newblock In {\em AAMAS},  407--414.

\bibitem[\protect\citeauthoryear{Pitt \bgroup et al\mbox.\egroup
  }{2006}]{DBLP:journals/cj/PittKSA06}
Pitt, J.; Kamara, L.; Sergot, M.~J.; and Artikis, A.
\newblock 2006.
\newblock Voting in multi-agent systems.
\newblock {\em Comput. J.} 49(2):156--170.

\bibitem[\protect\citeauthoryear{Popescu}{2013}]{DBLP:conf/hci/Popescu13b}
Popescu, G.
\newblock 2013.
\newblock Group recommender systems as a voting problem.
\newblock In {\em HCI (26)},  412--421.

\bibitem[\protect\citeauthoryear{Satterthwaite}{1975}]{Satterthwaite75}
Satterthwaite, M.
\newblock 1975.
\newblock Strategy-proofness and {A}rrow's conditions: Existence and
  correspondence theorems for voting procedures and social welfare functions.
\newblock {\em Journal of Economic Theory} 10(2):187--216.

\end{thebibliography}
